\documentclass[a4paper,conference]{IEEEtran}
\IEEEoverridecommandlockouts
\usepackage{epsf}
\usepackage{graphicx}
\usepackage[outdir=./]{epstopdf}
\usepackage{amsmath,bm}
\usepackage{amssymb}
\usepackage{epsfig,latexsym,amsmath,epsf,amssymb,amsfonts}
\usepackage{verbatim}

\usepackage[english]{babel}
\usepackage[utf8]{inputenc}
\usepackage{algorithm}
\usepackage[noend]{algpseudocode}

\usepackage{epstopdf}
\usepackage{amsthm}
\usepackage{url}
\usepackage[noadjust]{cite}
\usepackage{caption}
\usepackage{subcaption}

\newtheorem{theorem}{Theorem}

\usepackage{authblk}
\usepackage{amscd,mathrsfs}
\usepackage[numbers,sort]{natbib}
\usepackage[withbib,all]{authorindex}
\usepackage{esint}
\usepackage{paralist}
\usepackage{color}
\usepackage{bm}
\usepackage{amsthm}
\usepackage{mathtools}
\usepackage{mathptmx} 
\usepackage{times} 
\usepackage{amssymb}
\usepackage{dsfont}

\begin{document}

\title{Dynamic Content Updates in Heterogeneous Wireless Networks 
}
\author[1]{Mehdi Salehi Heydar Abad \thanks{This work was in part supported by EC H2020-MSCA-RISE-2015 programme under grant number 690893.}}
\author[2]{Emre Ozfatura}
\author[1]{Ozgur Ercetin}
\author[2]{Deniz G\"und\"uz}
\affil[1]{Faculty of Engineering
and Natural Sciences, Sabanci University}
\affil[2]{Department of Electrical and Electronic Engineering
Imperial College London}
\affil[1] {\textit{\{mehdis,oercetin\}@sabanciuniv.edu}}
\affil[2] {\textit{\{m.ozfatura,d.gunduz\}@imperial.ac.uk}}
\maketitle

\newtheorem{lemma}{Lemma}
\newtheorem{corollary}{Corollary}
\thispagestyle{empty}

\begin{abstract}
Content storage at the network edge is a promising solution to mitigate the excessive traffic load due to on-demand streaming applications as well as to reduce the streaming delay. To this end, cache-enabled cellular architectures can be utilized to increase the provided quality-of-service (QoS) and to reduce the network cost. However, there are certain issues to be considered in the design of the content storage strategy such that the contents should be refreshed in order to responds user`s expectations. Using a frequent cache refreshment strategy the ratio of satisfied users can be increased at an increasing network cost. In this paper, we introduce a cache refreshment strategy via leveraging learning techniques so that users' tolerance to the age of content is learned and the content is refreshed accordingly.  
\end{abstract}

\begin{IEEEkeywords}
Content caching, content refreshment, Markov decision process (MDP), quality of service, multi-armed bandit (MAB)
\end{IEEEkeywords}

\section{Introduction}

While proactive content caching has received significant interest in the recent years, most of the existing strategies in the literature (both with uncoded \cite{SC.CoUcCD1,SC.CoUcCD2,SC.CoUcCD3} and coded placement \cite{SC.CoCC1,SC.M3}) have been designed under the assumption that content popularities are known in advance. Although it is possible to observe the global popularity of contents in on-demand video streaming services, such as YouTube \cite{vs1},  small-cell base stations (SBSs) usually serve a small geographical area, where the local content popularity might not be aligned with the global popularity \cite{vs2}. This mismatch between the local and global content popularities requires the design of predictive caching policies that aim to learn the local content popularity from the user requests.
Predictive caching policies can be classified into two main groups, namely  {\em predictive caching with unknown popularities} \cite{SC.PCUP1, Blasco:ISIT:14, SC.PCUP2} and {\em predictive caching with time-varying popularities} \cite{SC.PCTP1,SC.PCTP2, Bharath:TC:18}. In \cite{SC.PCUP1, Blasco:ISIT:14}, the authors focus on a single SBS and model the predictive caching problem as a multi-arm bandit problem, in which the received user requests are utilized to predict the file popularities, and the optimal caching strategy is obtained by taking into account the cost of file replacements. In \cite{SC.PCUP2},  this approach has been extended to a cooperative caching framework, where the SBSs fetch the requested content from the neighboring SBSs if the corresponding  file is cached there. Another strategy for predictive caching follows a user-centric approach, possibly implemented at a higher layer, and records user requests for different contents as a matrix. Future user requests are predicted  using matrix completion techniques by exploiting the correlations among the requests for different files, similarly to recommendation systems \cite{SC.PCUP_MC}. In \cite{SC.PCTP1}, a cache replacement strategy  has been introduced for time-varying popularity scenario to maximize the local service rate with a minimum replacement cost, while a more theoretical approach is taken in \cite{Bharath:TC:18}, which studies the cache update policy in the case of time-varying content popularities. In parallel to the aforementioned works, another relevant research direction is {\em contextual predictive caching}, where various features of the requests, such as genre of the video file, age of the user, or the time of the requests, are utilized to predict future requests and shape the caching strategy accordingly \cite{SC.PC_cont1,SC.PC_cont2}.
Although the predictive caching framework is highly effective to increase the efficiency of edge caching, there are certain limitations. Most of the aforementioned predictive caching strategies are designed to predict only the content popularity. However, in many applications, e.g., news, weather, etc., freshness of the content is another important factor for the user satisfaction. Content caching and refreshment problem has been previously studied in  \cite{fresh1, fresh2}. In this paper, we consider a heterogeneous cellular network with cache enabled SBSs and provide a cost aware content update policy. The proposed content update strategy consists of two parts;  in the first part, we show that the structure of the optimal periodic content update policy that minimizes the network cost for given users' tolerance to the age of the content is of threshold type, and in the second part, we utilize the multi-armed bandit approach to learn users' tolerance to the age of  the content. To the best of our knowledge, ours is the first work that utilizes the learning framework to analyze the user behavior based on the content age. Accordingly, we design an online cache refreshment policy to minimize the overall network cost.

\begin{figure}
\centering
\includegraphics[scale=0.3]{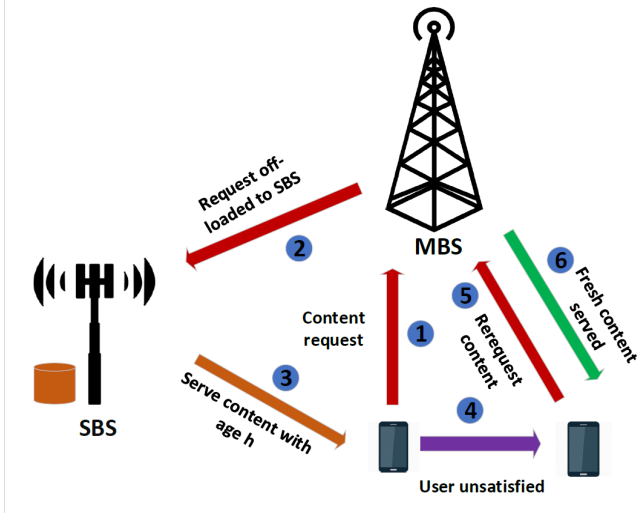}
\label{fig:sysmod}
\caption{User requests are first off-loaded to the SBS; however, users may become unsatisfied with the freshness of the contents they receive from the SBS, and consequently, such requests are re-directed to the MBS.}
\end{figure}
\section{System model and  problem formulation}

We consider a cellular network with a macro base station (MBS) and a SBS serving the users in a cell. It is assumed that both the MBS and the SBS are equipped with cache memories, storing a library of $N$ distinct dynamic contents, denoted by $S_{1}, \ldots, S_{N}$. We note that each dynamic content has a different popularity, i.e., the probability that a user requests content $S_{n}$ is $p_{n}$, for $n=1\ldots,N$. 

\subsection{Content freshness}
Dynamic contents, such as news videos, traffic and weather updates, may change frequently over time. While we assume that the MBS always has the fresh content updates, thanks to its relatively higher-bandwidth connection to the content server in the core network, the SBS needs to regularly refresh the dynamic contents in its cache to keep them up-to-date. For the SBS, downloading all the contents from the MBS through its limited backhaul link is costly in terms of energy, time and spectrum.

We consider a discrete time system model with equal-length time slots. At the beginning of each time slot the SBS decides on which contents to be updated. We assume that a content is updated at the beginning of the next time slot. The corresponding decision vector is denoted by  $\mathbf{d}(t)= \left(d_{1},\ldots,d_{N}\right)$, i.e., when the application $S_{n}$ is updated at the end of time slot $t$, $d_{n}(t)=1$, and  $d_{n}(t)=0$ otherwise. 

We denote by $h_{n}(t)$ the {\em age} of the content $S_{n}$ in the SBS cache at time slot $t$. We assume a maximum age $T_{max}$ at which a content becomes obsolete. In other words, the age of a content increases until it becomes obsolete. Accordingly, age of content $S_{n}$, $n=1,\ldots,N$, evolves over time in the following way:
\small
\begin{align}
h_{n}(t+1)=& \max \{(1-d_{n}(t))(h_{n}(t)+1),T_{max}\}. \label{eq:age}
\end{align}
\normalsize
We denote the length-$N$ vector of ages associated with all the dynamic contents in the library by $\mathbf{h}(t)$.

\subsection{User behavior}
Let $\lambda(t)$ denote the number of users that request a content at time slot $t$. Whenever a user requests a content, the request  is first off-loaded to the SBS to be served. Users have different tolerance levels to the age of the contents they receive. Hence, we consider that, with probability $P_{redirect}^{(n)}(h)$, a user is not satisfied with the age $h$ of content $S_n$, and thus, it places another request for content $S_{n}$. In that case, the new user request is served directly by the MBS with a fresh content. Let $\lambda_n(t)$ be the number of users that request content $n$, $n=1,\ldots,N$, in time slot $t$, which is governed by the popularity profile $p_n$. The number of users that request content $n$ is split into two disjoint sets, where the first set of users are redirected to the MBS, while the second set consists of the users satisfied with the service provided by the SBS. We denote these numbers by $\lambda_{rn}(t)$, $\lambda_{an}(t)$, respectively. Note that $\lambda_{rn}(t)$ and $\lambda_{an}(t)$ are governed by the random process $P_{redirect}^{(n)}(h_n(t))$. 

Let $\boldsymbol{\lambda}_r (t)= (\lambda_{r1}(t),\ldots,\lambda_{rN}(t))$ be the vector associated with the number of redirected users for each content. We have the following:
\small
\begin{align}
&\sum^{N}_{n=1} \lambda_n(t) = \lambda(t),\\
\text{where,}\hspace{0.5cm}&\lambda_{nr}(t)+\lambda_{na}(t) = \lambda_n(t).
\end{align}
\normalsize
The corresponding expected values of these parameters are given as:
\small
\begin{align}
&\mathds{E}[\lambda_n|\lambda] = \lambda p_n,\\
&\mathds{E}[\lambda_{nr}|\lambda,h_n] = \lambda p_n P_{r}^{(n)}(h_n).
\end{align}
\normalsize
\subsection{Decision model and the problem formulation}
Let $C(\boldsymbol{\lambda}_r(t),\mathbf{d}(t))$ be the cost associated with serving the users redirected to the MBS at time $t$, and the backhaul cost associated with updating the contents, if there is any. In this work, we assume that this cost is linear in  $\lambda_{rn}$\footnote{For example, in OFDMA, a user re-directed to the MBS is assigned a subcarrier, and the power allocated to that subcarrier adds linearly to the energy cost.}. Hence, 
\small
\begin{align}
C(\boldsymbol{\lambda}_r(t),\mathbf{d}(t)) = \sum^{N}_{n=1} C_n(\lambda_{rn}(t),d_n(t))\label{eq:cost_function}.
\end{align}
\normalsize

Similarly, we define the back-haul cost function $C_{BH}(\mathbf{d}(t))$ which is also a linear function.
\small
\begin{align}
C_{BH}(\mathbf{d}(t)) = \sum^{N}_{n=1}C^{(n)}_{BH}(d_n(t)),
\end{align}
\normalsize
where $C^{(n)}_{BH}(d_n(t)) = d_n(t)\mathcal{E}_n$, with $\mathcal{E}_n$ being the average back-haul cost of updating application $n$.
If the SBS decides to update the content $n$ (or multiple contents), the age of the content is updated at the end of the time slot. Hence,  
\small
\begin{align}
C_n(\lambda_{rn}(t),d_n(t)) = \beta_n + \alpha_n \lambda_{rn}(t) + d_n(t)\mathcal{E}_n
\end{align}
\normalsize
Note that updating a content has an immediate cost which is larger than not-updating. However, the incurred extra cost in updating the content enables more users to be served at a local SBS.

We aim at minimizing the average total cost as follows:
\small
\begin{align}
\min_{\mathbf{d}(t)\in\{0,1\}^{N}} \lim_{T\rightarrow\infty}\frac{1}{T}\sum^T_{t=1} \mathds{E}\big[C(\boldsymbol{\lambda}_r(t),\mathbf{d}(t))\big]
\end{align}
\normalsize

\section{MDP formulation}
Define the state of the system to be $\mathbf{h}$. We denote by $V(\mathbf{h})$ the differential value function at state $\mathbf{h}$. The differential Bellman equations can be written as:
\small
\begin{align}
V(\mathbf{h}) + \mu^* = \min_{\mathbf{d}}V_{\mathbf{d}}(\mathbf{h}),\label{eq:V_vec}
\end{align}
\normalsize
where $\mu^*$ is the optimal average cost and $V_{\mathbf{d}}(\mathbf{h})$ is the differential action-value function defined by:
\small
\begin{align}
V_{\mathbf{d}}(\mathbf{h}) = \bar{C}(\mathbf{h},\mathbf{d}) + \mathds{P}(\acute{\mathbf{h}}|\mathbf{h},\mathbf{d})V(\acute{\mathbf{h}}),
\end{align}
\normalsize
where $\mathds{P}(\acute{\mathbf{h}}|\mathbf{h},\mathbf{d})$ is the transition probability from state $\mathbf{h}$ into $\acute{\mathbf{h}}$ when action $\mathbf{d}$ is taken which is governed by \eqref{eq:age}, and
\small
\begin{align}
\bar{C}(\mathbf{h},\mathbf{d}) = \sum_{\boldsymbol{\Lambda}}\mathds{P}(\boldsymbol{\lambda_r} = \boldsymbol{\Lambda})C(\boldsymbol{\lambda}_r,\mathbf{d})
\end{align}
\normalsize
The MDP associated with the average cost minimization problem can be solved by well known value iteration algorithm. However, the cardinality of state space (i.e., $(T_{max})^{ N}$) and action space (i.e., $2^N$) grow exponentially with the number of contents. Hence, the curse of dimensionality is the bottleneck for an efficient solution. To bypass this bottleneck, we note that the cost function in \eqref{eq:cost_function} is linear and the transition probabilities of each content does not affect the other. Hence, we can separate the value function in \eqref{eq:V_vec} into $N$ independent value functions each representing a distinct application. For each application $n$, we have
\small
\begin{align}
&V^{(n)}(h_n) + \mu^*_n = \min_{d_n} V^{(n)}_d(h_n)\\
&V^{(n)}_d(h_n) = \bar{C}_n(h_n,d_n) + d_n V^{(n)}(0) + (1-d_n)V^{(n)}(h_n+1) \label{eq:bellman_content}
\end{align}
\normalsize

We have developed a framework that has enabled distributed policies with respect to the individual contents. We will show that for each content, there exists a threshold policy on the age of the content for which it is optimal to update the content. The following lemma establishes the key property used to prove the structure of the optimal policy.

\begin{lemma}
The differential value function $V^{(n)}(h_n)$ for all $n=1,\ldots,N$ is non-decreasing with respect to the age of the content, $h_n$.
\end{lemma}
\begin{proof}
We use the value iteration algorithm to prove the lemma. We start by an arbitrary $V^{(n)}_0(h)$ differential value function and obtain the $k$-step differential value function $V^{(n)}_k(h)$ as follows:
\small
\begin{align}
&V^{(n)}_{k+1}(h) \nonumber\\
&= \min_{d}(-\mu^*_n + \bar{C}_n(h_n,d_n) + d_n V^{(n)}_k(0) + (1-d_n)V^{(n)}_k(h+1))\label{eq:V_k}
\end{align}
\normalsize
Note that $\lim_{k\rightarrow\infty}V^{(n)}_k(h)=V^{(n)}(h)$. The proof is by induction. For $k=1$, $V^{(n)}_{1}(h) = \min_{d_n}(-\mu^*_n + \bar{C}_n(h_n,d_n) + d_n V^{(n)}_0(0) + (1-d_n)V^{(n)}_0(h_n+1))$ which is the minimum of two non-decreasing functions, and thus, itself is a non-decreasing function in $h_n$. Assume that the lemma holds for $k$. Then according to \eqref{eq:V_k}, $V^{(n)}_{k+1}(h_n)$ is also a non-decreasing function with respect to $h_n$. By letting $k\rightarrow\infty$, we conclude the proof by showing that $V^{(n)}(h_n)$ is also a non-decreasing function in $h_n$.
\end{proof}
The lemma is intuitively clear considering the non-decreasing property of the cost functions and Bellman equations in \eqref{eq:bellman_content}.

\begin{theorem}
For each content the optimal policy minimizing the  average cost is a threshold policy.
\end{theorem}

\begin{proof}
The monotonicity of the differential value functions prove the optimality of the threshold policy \cite[Chapter~7]{book:dp1}. Intuitively, due to the non-decreasing property of the  $V^{(n)}(h_n)$, at some age, it would be optimal to update the content. Since the differential value function is non-decreasing, a larger, or smaller age would not be able to yield a smaller average cost. 
\end{proof}

\section{Learning Content Popularity and Age Tolerance}
In the previous section, we showed that the problem is separable and thus, the optimization can be performed for each content separately. Second, we proved that the policy minimizing the cost is a threshold policy. Hence, the SBS by monitoring the age of the contents individually, needs to optimize according to a single threshold for each content. Under the threshold policy, the age of a content increases linearly until it reaches the threshold wherein the content will be updated and the age will refresh to a value of zero. Thus the minimum cost associated with content $n$ is the solution of:
\small
\begin{align}
\mu^*_n =\min_{H_n} \frac{1}{H_n+1}\bigg(\sum^{H_n}_{h=0}\bar{C}_n(h,0)+\mathcal{E}_n\bigg), \label{eq:opt_cost}
\end{align}
\normalsize
where
\small
\begin{align}
\bar{C}_n(h,0)=\beta_n + \alpha_n \lambda p_n P^{(n)}_{redirect}(h).
\end{align}
\normalsize
Considering the linearity of the cost functions, the average cost optimization becomes:
\small
\begin{align}
\min_{H_n}\left\{ \beta_n + \frac{\mathcal{E}_n + \alpha_n p_n \sum^{H_n}_{h=0}P^{(n)}_{redirect}(h)}{H_n + 1}\right\}.
\end{align}
\normalsize
The equivalent optimization problem depends on the redirection probabilities $P^{(n)}_{redirect}(h)$, that are unknown. Hence, in the following we resort to reinforcement learning methods to infer the redirection probabilities. 

We consider a sequential learning framework in which the SBS at each iteration of the learning algorithm faces choosing a threshold $0\leq H_{n}\leq T_{max}$. After choosing the threshold, the SBS will observe a random cost associated with its decision;
\small
\begin{align}
\hat{C}_n(H_n) = \frac{\mathcal{E}_n + \sum^{H_n}_{t=1}C_n(\lambda_{rn}(t),0)}{H_n + 1}
\end{align}
\normalsize
The learning algorithm should provide the SBS a method to adjust its strategy by observing the outcomes of its decision. This resembles the well-known multi-armed bandit (MAB) problem. In MAB, each action (i.e., thresholds) has an expected return value which is called \emph{value} of that action. We denote the true value of action $H_n$ by $q_n(H_n) = \frac{1}{H_n+1}\bigg(\sum^{H}_{h_n=0}\bar{C}_n(h,0)+\mathcal{E}_n\bigg)$. If the agent (i.e., SBS) knows $q(.)$ values, then it can simply choose the action with the minimum expected cost. Thus, we need an algorithm that can learn $q_n(.)$ values. A well-studied algorithm for learning those values is the $\epsilon$-greedy algorithm \cite{sutton} which starts by an arbitrary estimate, $Q_n(.)$ about the value of the actions and interacts with the environment to update its initial estimates, eventually converging to the true estimates. Two critical phases associated with $\epsilon$-greedy algorithm is the exploitation and exploration stages. The agent utilizing the estimates greedily chooses an action, and thus, it exploits what it knows already. Meanwhile, if it chooses an action completely random regardless of the estimates we say that it explores. Exploitation is necessary to act upon the experience while exploration helps to improve the estimate values and it facilitates convergence to the true action values. The $\epsilon$-greedy algorithm is presented in Algorithm \ref{alg:greedy}.

\begin{algorithm}\footnotesize{
\caption{$\epsilon$-greedy}\label{alg:greedy}
\begin{algorithmic}[1]
\For{$i = 1,2,\ldots$}
\State $
H_n\leftarrow\Bigg\{
\begin{array}{ll}
         \arg\max_{H_n} Q_n(H_n) & \mbox{with probability $1-\epsilon$},\\
        \mbox{random action} & \mbox{with probability $\epsilon$} .
        \end{array}
$
\State Apply $H_n$ and observe $\hat{C}_n(H_n)$
\State \small$Q_n(H_n)\leftarrow (1-\zeta)Q_n(H_n) + \zeta\hat{C}_n(H_n)$\normalsize
\EndFor
\end{algorithmic}}
\end{algorithm}

\section{Numerical Results}
In this section, we aim at evaluating the performance of the $\epsilon$-greedy algorithm in finding the optimal thresholds that minimize the total cost of the system. Due to the separability, we  consider only one content and we note that the learning processes for all the contents are the same. The popularity of the contents are modeled by a Zipf distribution with an exponent of $1.1$. We assume that on average a  given user becomes dissatisfied with a content of age $h$ with probability of $e^{-0.4h}$. Users arrive at the system according to a Poisson distribution with rate $100$ users per time slot. The cost of re-directed users to MBS is $C_n(\lambda_{rn}(t)) = 10\cdot\lambda_{rn}(t)$ and the backhaul cost is assumed to be $500$. In Figure \ref{fig:regret1}, we illustrate the performance of the $\epsilon$-greedy  algorithm for $\epsilon=0,\ 0.05,\ 0.1$ by adopting average regret as the metric. The regret of a learning algorithm is defined to be the difference between the cost achieved by the learner and the optimal cost. Here, we obtain the optimal cost by assuming that $P_{redirect}(h)$ is known, and by numerically solving \eqref{eq:opt_cost}. Note that the estimates of the action values, $Q(H)$, is initialized to be $0$ for all $H=0,\ldots,T_{max}$. Note also that, the greedy algorithm seems to achieve a better performance even if it always exploits. This is not too surprising considering that the action-values are initialized opportunistically (i.e., the initial costs are believed to be zero by the agent). At the beginning, i.e., $t=0$, the greedy algorithm believes that every action returns a value of zero. However, by trying each action it gets disappointed in that action and tries the rest. In other words the estimates are biased. Opportunistic initialization is a simple method to incentivize exploration. However, it can only happen once and at the beginning of time. This method quickly fails in non-stationary environments. 
 \begin{figure}[h]
  \centering
    \includegraphics[scale=.4]{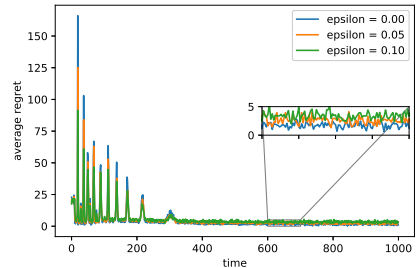}
		  \caption{Average regret of the $\epsilon$-greedy algorithm for $\epsilon=0,0.05,0.1$ in a stationary environment.}
			\label{fig:regret1}
\end{figure}

\begin{figure}
\includegraphics[scale=0.4]{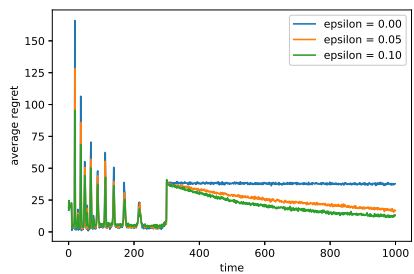}
\caption{Average regret of the $\epsilon$-greedy algorithm for $\epsilon=0,0.05,0.1$ in a non-stationary environment.}
\label{fig:regret2}
\end{figure}
To show this, we also study the performance of the $\epsilon$-greedy algorithm in a time varying environment. We assume that at $t=300$ the backhaul cost decreases to a value of $400$. The results are depicted in Figure \ref{fig:regret2}. We can see that the greedy algorithm cannot adapt to the non-stationary environment and it gets stuck in a sub-optimal threshold. Meanwhile, for $0.05$ and $0.1$-greedy algorithm, it is able to adapt to the environment thanks to their exploration strategy. A large value of $\epsilon$ results in more exploration, and thus, we can see that $0.1$-greedy algorithm has a faster decay in terms of the average regret compared to the $0.05$-greedy algorithm. However, note that $\epsilon$-greedy algorithm is expected to be at least $\epsilon$ away from the optimal cost. Thus, there  is a trade-off between the rate of convergence and the value of convergence.

\section{Conclusion}
In this work, we developed a framework for a cost minimization problem in a dynamic content caching setting. Specifically, we aimed at striking a balance between the number of unsatisfied users whom are redirected to MBS and the cost of accessing the backhaul link by SBS in updating the dynamic contents. We formulated the cost minimization problem as an MDP and showed that the problem is separable with respect to the contests. Subsequently, we proved that a threshold policy in the age of the contents is optimal. In finding the optimal thresholds, we resorted to learning algorithms since users' preferences are not known and vary with each content. To that extend, we represented the problem in MAB framework and through numerical results, we showed that it is possible to make the expected regret of the learning algorithm arbitrarily close to zero. The learning algorithm even shows adaptability to non-stationary settings. As a future work, we aim to investigate the system with non-linear cost functions. Under the non-linear cost functions, the problem is no longer separable and new solution methods needs to be investigated. As an extension of this work we will also analyze the heterogeneous cellular network architecture with energy harvesting SBSs.

\bibliographystyle{IEEEtran}
\bibliography{ref}
\end{document}